\documentclass[twocolumn]{article}
\pdfoutput=1
\usepackage{graphicx,amsmath,amsfonts,amssymb,amscd,bezier,amstext,makeidx}
\usepackage[all]{xy}
\usepackage{makeidx}
\makeindex
\usepackage[latin1]{inputenc}
\usepackage[english]{babel}
\usepackage{srcltx}
\usepackage{color}
\usepackage{epsfig}
\usepackage[ruled,vlined]{algorithm2e}  % para el pseudo-codigo
\usepackage{tikz}

\usepackage{amsthm}
\usepackage{hyperref} %hiper-referencias, claro!
\usepackage[cm]{fullpage}%magic on margins %-------------- COMMENT for IEEE ----------------

\newtheorem{proposition}{Proposition}
\newtheorem{definition}[proposition]{Definition}

\newtheorem{theorem}[proposition]{Theorem}
\newtheorem{corollary}[proposition]{Corollary}
\newtheorem{conjecture}[proposition]{Conjecture}
\newtheorem{lemma}[proposition]{Lemma}
\newtheorem{remark}[proposition]{Remark}

\newtheorem{example}{Example}

\begin{document}

\title{Título}
\author{Autores}
\title{Quasi-perfect Lee Codes of Radius 2 and Arbitrarily Large Dimension}
\author{Cristóbal Camarero and Carmen Martínez\\
Department of Computer Science and Electronics\\
Universidad de Cantabria, UNICAN, Spain.}
% \thanks{This work has been supported by the Spanish Science and Technology Commission (CICYT) under contracts TIN2013-46957-C2-2-P and AP2010-4900.
% 
% C. Camarero and C. Mart\'{\i}nez are with the Department of Computer Science and Electronics, Universidad de Cantabria, UNICAN, Spain.
% email: cristobal.camarero@unican.es; carmen.martinez@unican.es;}
\date{Online version in IEEE Transactions on Information Theory: \href{http://dx.doi.org/10.1109/TIT.2016.2517069}{doi:~10.1109/TIT.2016.2517069}}

\maketitle
%%% CARMEN: Cambiado lo de infinitely such primes -> infinitely many such primes
%%% cambiado Welch conjectured that perfect codes for the Lee-metric do not to exist for dimension -> Welch conjectured that perfect codes for the Lee-metric do not exist for dimension
\begin{abstract} A construction of 2-quasi-perfect Lee codes is given over the space $\mathbb Z_p^n$ for $p$ prime, $p\equiv \pm 5\pmod{12}$ and $n=2[\frac{p}{4}]$. It is known that there are infinitely many such primes. Golomb and Welch conjectured that perfect codes for the Lee-metric do not exist for dimension $n\geq 3$ and radius $r\geq 2$. This conjecture was proved to be true for large radii as well as for low dimensions. The codes found are very close to be perfect, which exhibits the hardness of the conjecture. A series of computations show that related graphs are Ramanujan, which could provide further connections between Coding and Graph Theories.
\end{abstract}

% \begin{IEEEkeywords}
% Lee metric; quasi-perfect codes; Cayley graphs; Gaussian integer ring.
% \end{IEEEkeywords}

\section{Introduction}\label{sec1}

%%% CARMEN: quitada la footnote y añadido For the cases in which a perfect code exists, Post determined -> For the cases in which a periodic perfect code exists, Post determined 
Golomb and Welch conjectured in their seminal paper \cite{Golomb} that perfect Lee codes only exist for spheres of radius $r = 1$ or in Lee spaces of dimension $n = 1, 2$. A constructive result for 1-perfect Lee codes was also given in that paper. Moreover, for a radius sufficiently greater than the space dimension, a negative existence result was obtained by approximating the problem to the densest tessellation of $\mathbb R^n$ with cross-polytopes. Afterwards, Moln\'{a}r enumerated all lattice-like 1-perfect codes in \cite{Molnar}. Later, Post in \cite{Post} gave a strong negative result. For the cases in which a periodic perfect code exists, Post determined an upper bound for its radius. In this upper bound the radius must fulfill $r<\frac{1}{2}n\sqrt{2}-\frac{3}{4}\sqrt{2}-\frac{1}{2}$ for $n\geq 6$. Later, J. Astola~\cite{Astola} and Lepist\"o~\cite{Lepisto} improved the bound given by Post to a quadratic relation between $r$ and $n$, which can be considered as an Elias-type bound for Lee codes. These negative results suggest that the conjecture is more difficult for radius 2, as was argued by Horak in \cite{Horak_n6}.\par

Other authors have considered the conjecture for low dimensions. For example, Gravier \textsl{et al.} in \cite{Gravier} proved the non-existence of perfect codes in 3-dimensional Lee spaces, even considering spheres of different radii. Dimension 4 was studied by \v{S}pacapan in \cite{Spacapan}, again with the possibility of spheres of different radii and all of them being greater or equal to 2. Also, Horak in \cite{Horak_n5} and \cite{Horak_n6} proved the non-existence of perfect Lee codes for $r=2$ and spaces of dimension $n = 5, 6$.
%This result together with those provided by Post prove that there are no perfect codes of any radius and dimension up to $n \leq  6$.
Later, Horak and Gro\v{s}ek in \cite{Horak_radius2} computationally proved the non-existence of perfect Lee codes for dimension $n\leq 12$ and radius $r=2$ by restricting the problem to linear codes.\par

In addition, several papers have considered problems involving the conjecture that could provide some insight about it. One approach has been to generalize the Lee metric. Huber in \cite{Huber_gaussian} gave 1-perfect codes over Gaussian integers and some non-perfect codes with greater correction. In \cite{Costa_tessellation} Costa \textsl{et al.} considered a relation between tessellations, graphs and codes over flat tori. In \cite{IEEE_TIT, ISIT08, IEEETITQuat} Martinez \textsl{et al.} gave a generalization of the Lee distance by means of a family of Cayley graphs over Cayley--Dickson algebras. Also, the existence of perfect codes being ideals of the algebras was considered. Nishimura and Hiramatsu in \cite{Nishimura} generalized the Lee distance using a surjective function from $\mathbb Z^l$ into a finite field and constructed some non-perfect 2-error correcting codes for this metric.

The existence of Lee codes has also been considered in terms of the size $q$ of the alphabet. AlBdaiwi \textsl{et al.} in \cite{AlBdaiwi_enumeration} enumerated all the alphabet sizes $q$ such that there exists a linear 1-perfect Lee code over $\mathbb{Z}_q^{n}$. In \cite{Astola_bounds} H. Astola and Tabus obtained, for small alphabet size $q$ and dimension $n$, an upper bound of the number of codewords of
error correcting Lee codes.

Recently, a new approach has been taken in terms of \textsl{diameter perfect codes}, which were introduced by Ahlswede \textsl{et al.} in \cite{Ahlswede}. A subset $\mathcal{C} \subseteq \mathbb{Z}_{q}^{n}$  is a diameter perfect code if there exists an anticode $\mathcal{A}$ such that $|\mathcal{C}||\mathcal{A}|= q^n$. This concept generalizes perfect codes since diameter perfect codes with minimum distance being odd are in fact the perfect codes. Etzion in \cite{Etzion_diameter} built diameter perfect codes of minimum distance 4. Later, Horak and AlBdaiwi~\cite{Horak_diameter} enumerated the arities $q$ such that there are 4-diameter perfect codes over $\mathbb Z_q^n$. Araujo \textsl{et al.} in \cite{Araujo} presented a generalization of diameter perfect Lee codes, together with a new conjecture that extends the conjecture by Golomb and Welch. Etzion \textsl{et al.} in \cite{Etzion_weighing} built Lee codes for large dimension by means of weighing matrices.

Another way of proving the existence of perfect Lee codes has been to relax the condition of being perfect. Thus, although not widely used, quasi-perfect codes for the Lee metric have been considered. AlBdaiwi and Bose in \cite{AlBdaiwi_quasiperfect} presented some quasi-perfect codes for dimension $n=2$. Also, in \cite{Horak_radius2} the authors presented some quasi-perfect codes for $n=3$ and a few radii. Later, Queiroz \textsl{et al.} in \cite{Queiroz} characterized quasi-perfect codes over Gaussian and Eisenstein--Jacobi integers being linear. As a consequence, linear quasi-perfect Lee codes were obtained for $n=2$.

%%CARMEN: large, small

In the present paper an explicit construction of linear quasi-perfect Lee codes of radius 2 for arbitrarily large dimensions. It will be shown that these codes are very close to be perfect, since they have half the density of potential perfect codes. By contrast, all other results to this date depend on the dimension of the space. As a consequence, combinatorial arguments may be insufficient to address the conjecture. Nevertheless, the relation with Cayley graphs studied in this paper indicates that the conjecture has also algebraic features. Moreover, in the authors opinion, the existence of these quasi-perfect codes, hints that a perfect code might exist for small radius, although this is contrary to the general believe.

These quasi-perfect 2-error correcting Lee codes will be defined by means of Cayley graphs over Abelian finite groups. The degree of the graph will be the double of the dimension of the Lee space. The order of the graph will be in inverse relation to the density of the quasi-perfect code. Thus, the main contribution of the paper is presented in the next result.

\begin{theorem}\label{theo:existencia} For any prime $p\geq 7$ such that $p\equiv\pm 5 \pmod{12}$ there exists a linear
2-quasi-perfect $p$-ary Lee code over $\mathbb{Z}^n_p$, where $n=2\left[\frac{p}{4}\right]$ and with $p^{n-2}$ codewords.
\end{theorem}

Note that the notation $[a]$ stands for the closest integer to the rational number $a$. As an example of the codes obtained in previous result, let us consider the following:

\begin{example}
Let $n = 4$, $p= 7$. Then, the code over $\mathbb{Z}^{4}_{7}$ defined by the parity-check matrix

$$\left(\begin{array}{cccc}
1 & 0 & 2 & -2 \\
0 & 1 & 2 & 2 \\
\end{array}\right)$$ results in a 2-quasi-perfect 7-ary Lee code over $\mathbb{Z}^4_7$. This code has $p^{n-2} = 49$ codewords. It is known that perfect codes do not exist in this case since the sphere packing bound is $\frac{7^4}{41} \approx 58.56$.
\end{example}

As a consequence of Dirichlet's theorem on arithmetic progressions, there are infinitely many primes $p$ such that $p\equiv 5\pmod {12}$ and infinitely many primes such that $p\equiv -5\pmod {12}$. Thus, for any constant $c$, there is a prime $p\equiv \pm 5\pmod {12}$ such that the dimension $n=2\left[\frac{p}{4}\right]$ is greater than $c$. As a consequence of this and Theorem~\ref{theo:existencia}, it is obtained that:

\begin{corollary} There are infinitely many $n \in \mathbb{N}$ such that there exists a 2-quasi-perfect Lee code over a $n$-dimensional Lee space.
\end{corollary}

As it will be shown later, the result is constructive, and any application that requires the use of Lee-codes could benefit from it. For example, Roth and Siegel in \cite{Roth} considered BCH Lee codes and their application to constrained and partial-response channels. Using space embeddings, Jiang \textsl{et al.} in \cite{Jiang_CCRM} gave a method to construct Charge-Constrained Rank-Modulation codes (CCRM codes) from Lee error-correcting codes, which could be employed for flash memories. H. Astola and Stankovic in \cite{Astola_diagrams} considered Lee codes to build decision diagrams.

The rest of the paper is organized as follows. Since the codes considered in this paper will be defined by means of Cayley graphs, in Section~\ref{sec:grafo-codigo} the relation between Lee codes and Cayley graphs over Gaussian integers is stated. Moreover, the family of Cayley graphs under study is defined. Then, in Section~\ref{sec:correcion} the Cayley graphs selected are proved to have error correction capacity 2. In Section~\ref{sec:diametro} those Cayley graphs are shown to attain diameter 3, which implies that they define 2-quasi-perfect codes. Finally, in Section~\ref{sec:conclusiones} the results presented in this paper are discussed, and some open problems and future lines of research are detailed.

\section{Codes and Graphs}\label{sec:grafo-codigo}

Linear 2-quasi-perfect $p$-ary linear Lee codes are going to be defined by means of Cayley graphs. Therefore, the correspondence between a linear code and a Cayley graph is explained in this section. First, some fundamental definitions are stated here.

Since Lee codes are the target of our study, the natural space to be considered is the finite integer lattice $\mathbb Z_p^n$. However, for convenience, also the infinite lattice $\mathbb Z^n$ will be considered.
Therefore, a \textsl{code} $\mathcal{C}$ will be a subset of either $\mathbb Z_p^n$ or $\mathbb Z^n$. This code is said to be \textsl{linear} if it is a subgroup of the corresponding space.

The \textsl{Manhattan distance} will be used in the space $\mathbb Z^n$. For any two words $v, w \in \mathbb Z^n$ its Manhattan distance is defined as:
	$$d(v,w)=\sum_{j=1}^n |v_i-w_i|.$$
On the other hand, the \textsl{Lee distance} will be the metric used when considering $\mathbb Z_p^n$. For $v, w \in \mathbb Z_p^n$ its Lee distance is defined as
	$$d(v,w)=\sum_{j=1}^n \min\{ |s| \mid s\equiv v_j-w_j\pmod p,\ s\in\mathbb Z\}.$$
Since the Lee distance becomes the Manhattan distance for $p=\infty,$ there will be no opportunity for confusion.
In both cases the weight of a word $v$ is defined as its distance to the origin, which will be denoted as $|v|=d(v,O)$.
For any positive integer $r$, the \textsl{Lee sphere} of radius $r$ is defined as all the points whose weight is less or equal to $r$, that is:
	$$B_r^n=\{v \mid |v|\leq r\}.$$
Note that, for any dimension $n \geq 1$, the cardinal $|B_2^n|=2n^2+2n+1,$ \cite{Golomb}.

A code $\mathcal C$ is said to be \textsl{$t$-error correcting} if $t$ is the greatest integer such that for any word $w$ there is at the most one codeword $c\in\mathcal C$ with $d(w,c)\leq t$. Thus, $t$ is called the \textsl{error correction} of $\mathcal C$. A code $\mathcal C$ is said \textsl{$r$-covering} if $r$ is the smallest integer such that for any word $w$ there is at least one codeword $c\in\mathcal C$ with $d(w,c)\leq r$. Thus, $r$ denotes the covering radius of $\mathcal C$. Then, a code that is both $t$-error correcting and $t$-covering, or equivalently with error correction equal to its covering radius, is said to be \textsl{perfect}. Golomb and Welch in \cite{Golomb} conjectured that there only exist perfect Lee codes for $t=1$ or $n=2$. Therefore, the existence of quasi-perfect codes must be studied since they are the best alternative to the perfect codes. Thus, a code that is $t$-error correcting and $(t+1)$-covering is said to be \textsl{$t$-quasi-perfect}. In this work $2$-quasi-perfect Lee codes are found for arbitrarily large dimensions. This is done by the construction of a family of Cayley graphs that leads to the codes definition. The remainder of this section is devoted to define this relation between codes and graphs. For simplicity, the infinite lattice $\mathbb{Z}^{n}$ will be considered and an equivalent relation can be stated in the case of $\mathbb{Z}_p^{n}$.\par

Given a group $\Gamma$ and a set of generators $H = \{\beta_1, \ldots , \beta_s\} \subset \Gamma$, the \textsl{Cayley graph} over $\Gamma$ generated by $H$ is defined as the graph with set of vertices the elements of $\Gamma$, and adjacencies $(u, u + \beta_i),$ for every $u \in \Gamma$ and $i = 1, \ldots , s$. $H$ must satisfy $H=-H$ and $0\not\in H$ in order to be a simple undirected graph. Since only Abelian groups will be considered, the operation of the group will be denoted by $+$ and the neutral element by 0. Now, given a linear code $\mathcal C\subset\mathbb Z^n$  the associated graph is
	$$G=Cay(\mathbb Z^n/\mathcal C ; \{\pm e_1,\dotsc,\pm e_n\}).$$

Reciprocally, given a Cayley graph $Cay(\Gamma ; \{\pm a_1,\allowbreak\dotsc,\allowbreak\pm a_n\})$ a linear code can be built. First, let us consider the homomorphism $\phi: \mathbb{Z}^n \longrightarrow\Gamma$ such that $\phi(e_j)= a_j$. Then, the code is given by
	$$\mathcal C=\{x\in\mathbb Z^n \mid \phi(x)=0\}=\ker \phi.$$

Next, distance and correction parameters of both the code and the graph are related as Theorem \ref{theo:grafo-codigo} states. Now, let us recall some basic definitions. The \textsl{distance} $d_G(v,w)$ between two vertices $v$, $w$ in a graph $G$ is defined as the number of edges in the shortest path from $v$ to $w$. Then, the \textsl{diameter} of a graph $G$ is the maximum among distances between every pair of vertices. Since Cayley graphs are vertex-transitive, this can also be calculated as the maximum distance to one particular vertex, usually $0 \in \Gamma$.

\begin{definition} Given a Cayley graph $Cay(\Gamma ;\allowbreak\{\pm a_1,\allowbreak\dotsc, \allowbreak\pm a_n\})$ over an Abelian group $\Gamma$, its \textsl{error correction capacity} is defined as the greatest integer $t$ such that for every vertex $v\in \Gamma$ there are $|B_t^n|$ vertices at distance $t$ or less from $v$.
\end{definition}

Note that since $G$ is a Cayley graph, it is vertex-transitive and therefore it is enough to count the number of vertices around one vertex to determine its error correction capacity. Thus, the equivalence between distance and covering properties of a linear Lee code and its associated Cayley graph over an Abelian group is proved in the following theorem:

\begin{theorem}\label{theo:grafo-codigo} Let $\Gamma$ be a finite Abelian group that is generated by $\{a_1,\dotsc,a_n\}$ and let $G=Cay(\Gamma ; \{\pm a_1,\dotsc, \pm a_n\})$.
Let $\phi$ be the homomorphism from $\mathbb Z^n$ into $\Gamma$ defined by $\phi(e_j)= a_j$ and let $\mathcal C=\ker \phi$. Then,
\begin{enumerate}

%diameter r
\item the diameter of $G$ equals the covering radius of $\mathcal C$ and
%error capacity t
\item the error correction capacity of $G$ equals the error correction of $\mathcal C$.

\end{enumerate}
\end{theorem}

But first of all, let us prove the following technical result.

\begin{lemma}\label{lemma:1}
In the hypothesis of Theorem \ref{theo:grafo-codigo}, for every $x\in \mathbb Z^n$ it is obtained that $d_G(\phi(x),0)=d(x,\mathcal C)$.
\end{lemma}

\begin{proof} Let $x$ be an arbitrary element of $\mathbb Z^n$.
Let us prove first that $d_G(\phi(x),0)\leq d(x,\mathcal C)$.
Let $c$ be the closest codeword to $x$, so $d(x,\mathcal C)=d(x,c)$.
By the definition of Manhattan distance, $d(x,c)=\sum_{j=1}^n|x_j-c_j|$.
As $\mathcal C$ is the kernel of $\phi$, $\phi(x)=\phi(x-c)=\sum_{j=1}^n(x_j-c_j)a_j$.
The distance in a Cayley graph over an Abelian group of a vertex $v$ to 0 is given by $d_G(v,0)=\min\{\sum_{j=1}^n |y_j|\mid \sum_{j=1}^n y_ja_j=v\}$.
Hence, taking $v=\phi(x)$ and $y=x-c$ in the previous expression, $d_G(\phi(x),0)\leq \sum_{j=1}^n |x_j-c_j|=d(x,\mathcal C)$.

For the second inequality, $d_G(\phi(x),0)\geq d(x,\mathcal C)$, let $y\in\mathbb Z^n$ be the vector such that $d_G(\phi(x),0)=\sum_{j=1}^n |y_j|$ and $\phi(x)=\sum_{j=1}^n y_ja_j$.
By definition of $y$, $\phi(x)=\phi(y)$, so the difference $c=x-y\in\mathcal C$ is a codeword.
Thus, as $\mathcal C$ is a linear code, $d(x,\mathcal C)\leq d(x,c)=d(x-c, O)=d(y,O)$, and by the definition of Manhattan distance $d(x,\mathcal C)\leq d(y,O)=\sum_{j=1}^n |y_j|=d_G(\phi(x),0)$.
\end{proof}

\begin{proof}(of Theorem \ref{theo:grafo-codigo}). For the first item in the theorem note that $\mathrm{diam}(G)
	=\max\{d_G(v,0)\mid v\in\Gamma\} =	\max\{d_G(\phi(x),0)\mid x\in\mathbb Z^n\},$ since $\phi$ is surjective. Then, by Lemma \ref{lemma:1} it is obtained that $\max\{d_G(\phi(x),0)\mid x\in\mathbb Z^n\} = \max\{d(x,\mathcal C)\mid x\in\mathbb Z^n \}	=\mathrm{covering\,radius}(\mathcal C)$.

For the second item let us proof that for every integer $s>0$, $ \phi(B_{s}^{n}) = \{ v \in \Gamma \mid d_{G}(v, 0) \leq s\}.$ If $x$ belongs to $B_s^n$ then $d(x, O) \leq s$ and, by Lemma \ref{lemma:1}, $\phi(x)$ is at distance at most $s$ from 0. Reciprocally, if there is a vertex $v \in \Gamma$ such that $d_G(v,0)\leq s$ then, there exists $x$ such that $\phi(x)=v$ and $x\in B_s^n$.

Now, let $t$ be the error correction of $\mathcal C$, that is, the greatest integer such that all the words in $B_t^n$ are closer to $O$ than to any other codeword. If $x,y\in B_t^n$ are such that $\phi(x)=\phi(y)$ then $\phi(x-y)=0$. Hence, $x-y\in \ker\phi=\mathcal C$. Since $d(x,y)\leq 2t$ it is obtained that $x=y$. Therefore, the cardinal numbers $|B_t^n|=|\phi(B_t^n)|$ are equal. Thus, by previous step $|\{ v \in \Gamma \mid d_{G}(v, 0) \leq t\}| = |B_t^n|,$ which implies that the error correction capacity of $G$ is at least $t$.

Finally, let us denote by $t'$ the error correction capacity of $G$. Again, there are $|B_{t'}^{n}|$ words at a distance less of equal to $t'$ from 0. Let $x\in\mathbb Z^n$ and let $c_1,c_2 \in \mathcal{C}$ be such that $d(x, c_1),d(x,c_2) \leq t'$. Then, note that $c_1-x, c_2-x \in B_{t'}^{n}$ and $\phi(-x) = \phi(c_1) + \phi(- x) = \phi(c_1-x)=\phi(c_2-x)$. From  $|B_{t'}^n|=|\phi(B_{t'}^n)|$, it is obtained that $\phi$ restricted to this set is a bijection, which implies that $c_1-x = c_2-x$, that is, $c_1=c_2$, which concludes the proof.
\end{proof}

\begin{remark}
Note that $Cay(\Gamma;\{\pm a_1,\dotsc, \pm a_n\})\cong Cay(\mathbb Z^n/\ker \phi ; \{\pm e_1,\dotsc,\pm e_n\})$. Thus, applying the previous procedure to obtain a code from a graph and applying it again to obtain a graph from a code, then an isomorphic graph is obtained.
\end{remark}

\begin{remark} Theorem \ref{theo:grafo-codigo} can be graphically interpreted by means of tessellations, as illustrated in Figure~\ref{fig:code-tessellation-graph}. Subfigure a) shows $\mathcal C = \langle (4, 4), (-4, 4) \rangle $, a 3-quasi-perfect linear Lee code over $\mathbb{Z}_{16}^{2}$. This is, the code has error correction 3 and covering radius 4. Subfigure b) shows a Voronoi tessellation induced by $\mathcal C$, in which every tile has as center a codeword. Subfigure c) shows in detail one of these tiles. As it can be observed, it contains the Lee sphere $B_3^2$ and it is contained in the Lee sphere $B_4^2$. Subfigure d) shows the Cayley graph $Cay(\frac{\mathbb{Z}_{16}^2 }{\mathcal{C}};\{\pm e_1,\pm e_2\})$. This graph is induced by tessellation as follows. The vertices are the words in the tile and two vertices are adjacent if they are at a distance 1, modulo the tessellation. Finally, observe that the graph has diameter 4 since there are 7 vertices at distance 4 from the center. Also, it has error correction capacity 3 since there are $25 = |B_{3}^2|$ vertices at a distance less or equal to 3.

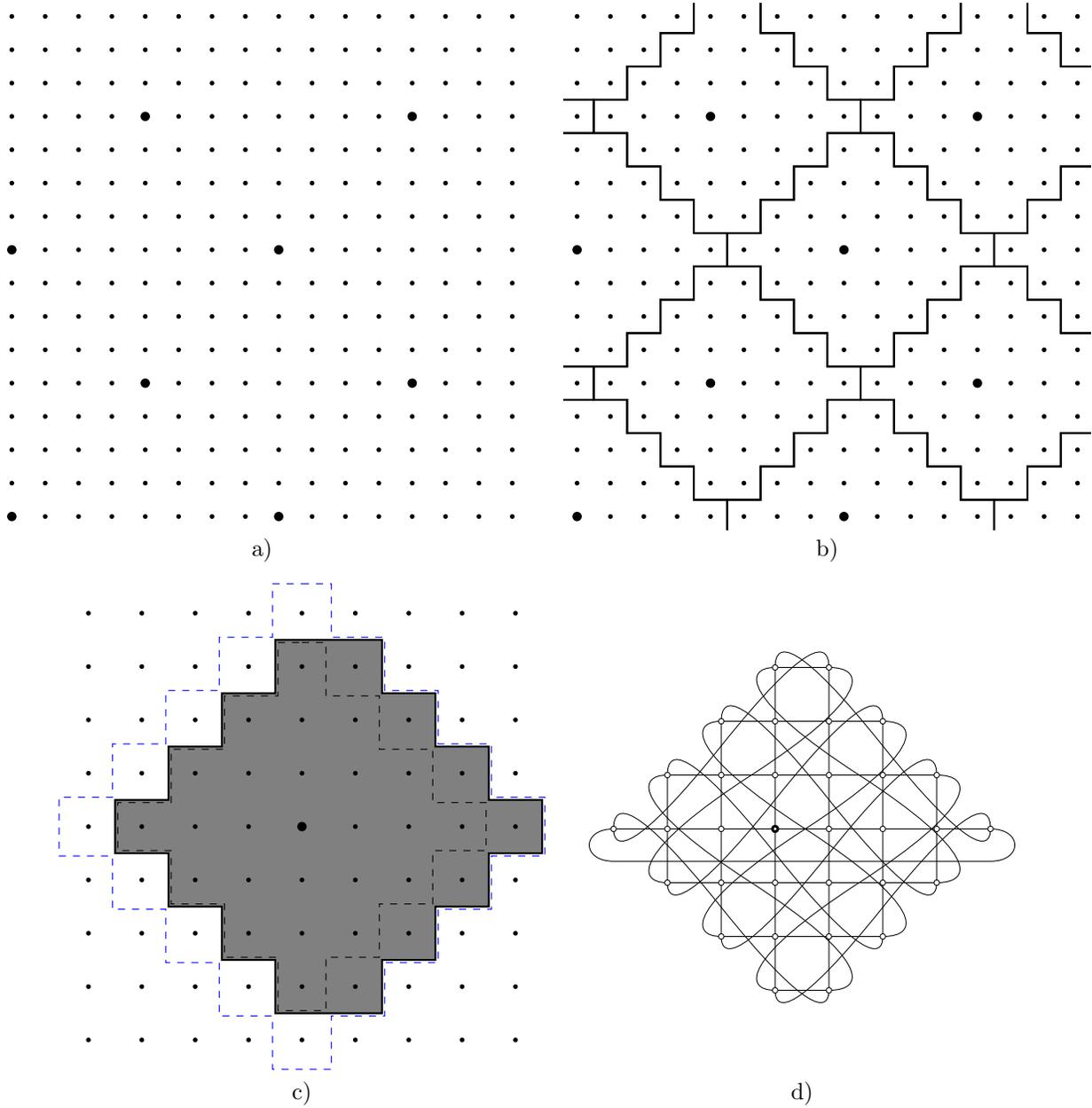
\begin{figure*}
	\begin{center}
	%\tikzsetnextfilename{3-quasi-perfect-code}
	%\tikzset{external/export next=true}
	%4+4i over 16x16
	\typeout{Drawing Figure}
	\begin{tikzpicture}[x=.5cm,y=.5cm]
		\foreach \x in {0,...,15}
		\foreach \y in {0,...,15}
		{
		\fill (\x,\y) circle (1pt);
		}
		\foreach \k in {0,...,3}
		\foreach \j in {0,...,1}
		{
			\pgfmathsetmacro\x{mod(16+4*\k-4*\j,16)}
			\pgfmathsetmacro\y{mod(4*\k+4*\j,16)}
			\fill (\x,\y) circle (2pt);
		}
		\node at (7.5,-1) {a)};
		\path (-0.4,-0.4) rectangle (15.4,15.4);
	\end{tikzpicture}
	\hskip 3ex
	% The Voronoi Tessellation
	\begin{tikzpicture}[x=.5cm,y=.5cm]
		\foreach \x in {0,...,15}
		\foreach \y in {0,...,15}
		{
		\fill (\x,\y) circle (1pt);
		}
		\begin{scope}
		%\clip (0,0) rectangle (15,15);
		\clip (-0.4,-0.4) rectangle (15.4,15.4);
		\foreach \k in {0,...,4}
		\foreach \j in {-2,...,3}
		{
			\pgfmathsetmacro\x{4*\k-4*\j}
			\pgfmathsetmacro\y{4*\k+4*\j}
			\fill (\x,\y) circle (2pt);
			%\draw[overlay,dashed] (\x,\y) +(0,3) -- +(3,0) -- +(0,-3) -- +(-3,0) -- +(0,3);%correction
			%\draw[overlay,dashed] (\x,\y)
			%	+(0,3.45) -| +(0.45,2.45) -| +(1.45,1.45) -| +(2.45,0.45) -|
			%	+(3.45,0) |- +(2.45,-0.45) |- +(1.45,-1.45) |- +(.45,-2.45) |-
			%	+(0,-3.45) -| +(-0.45,-2.45) -| +(-1.45,-1.45) -| +(-2.45,-.45) -|
			%	+(-3.45,0) |- +(-2.45,.45) |- +(-1.45,1.45) |- +(-.45,2.45) |- +(0,3.45);%correction. cubistic
			%\draw[overlay,thick] (\x,\y) +(0,3) -- +(1,3) -- +(4,0) -- +(1,-3) -- +(0,-3) -- +(-3,0) -- +(0,3);%compact
			\draw[overlay,thick] (\x,\y)
				+(0,3.5) -| +(1.5,2.5) -| +(2.5,1.5) -| +(3.5,0.5) -|
				+(4.5,0) |- +(3.5,-0.5) |- +(2.5,-1.5) |- +(1.5,-2.5) |-
				+(0,-3.5) -| +(-0.5,-2.5) -| +(-1.5,-1.5) -| +(-2.5,-.5) -|
				+(-3.5,0) |- +(-2.5,.5) |- +(-1.5,1.5) |- +(-0.5,2.5) |- +(0,3.5);%cubistic
			%\draw[overlay,thick] (\x,\y)
			%	+(0,3.4) -| +(1.4,2.4) -| +(2.4,1.4) -| +(3.4,0.4) -|
			%	+(4.4,0) |- +(3.4,-0.4) |- +(2.4,-1.4) |- +(1.4,-2.4) |-
			%	+(0,-3.4) -| +(-0.4,-2.4) -| +(-1.4,-1.4) -| +(-2.4,-.4) -|
			%	+(-3.4,0) |- +(-2.4,.4) |- +(-1.4,1.4) |- +(-0.4,2.4) |- +(0,3.4);%cubistic
		}
		\end{scope}
		\node at (7.5,-1) {b)};
	\end{tikzpicture}
	\vskip .5em
	% A tile
	\begin{tikzpicture}[x=.8cm,y=.8cm]
		\foreach \x in {-4,...,4}
		\foreach \y in {-4,...,4}
		{
		\fill (\x,\y) circle (1pt);
		}
		\fill (0,0) circle (2pt);
		%\draw[dashed] (0,3) -- (3,0) -- (0,-3) -- (-3,0) -- (0,3);%correction. compact
		%\draw[dashed]
		%	(0,3.5) -| (0.5,2.5) -| (1.5,1.5) -| (2.5,0.5) -|
		%	(3.5,0) |- (2.5,-0.5) |- (1.5,-1.5) |- (.5,-2.5) |-
		%	(0,-3.5) -| (-0.5,-2.5) -| (-1.5,-1.5) -| (-2.5,-.5) -|
		%	(-3.5,0) |- (-2.5,.5) |- (-1.5,1.5) |- (-.5,2.5) |- (0,3.5);%correction. cubistic
		\draw[dashed]
			(0,3.45) -| (0.45,2.45) -| (1.45,1.45) -| (2.45,0.45) -|
			(3.45,0) |- (2.45,-0.45) |- (1.45,-1.45) |- (.45,-2.45) |-
			(0,-3.45) -| (-0.45,-2.45) -| (-1.45,-1.45) -| (-2.45,-.45) -|
			(-3.45,0) |- (-2.45,.45) |- (-1.45,1.45) |- (-.45,2.45) |- (0,3.45);%correction. cubistic
		%\draw[thick] (0,3) -- (1,3) -- (4,0) -- (1,-3) -- (0,-3) -- (-3,0) -- (0,3);%compact
		\draw[thick,fill=black,fill opacity=0.5]
			(0,3.5) -| (1.5,2.5) -| (2.5,1.5) -| (3.5,0.5) -|
			(4.5,0) |- (3.5,-0.5) |- (2.5,-1.5) |- (1.5,-2.5) |-
			(0,-3.5) -| (-0.5,-2.5) -| (-1.5,-1.5) -| (-2.5,-.5) -|
			(-3.5,0) |- (-2.5,.5) |- (-1.5,1.5) |- (-0.5,2.5) |- (0,3.5);%cubistic
		%\draw[dashed,blue] (0,4) -- (4,0) -- (0,-4) -- (-4,0) -- (0,4);%diameter. compact
		%\draw[dashed,blue]
		%	(0,4.5) -| (0.5,3.5) -| (1.5,2.5) -| (2.5,1.5) -| (3.5,.5) -|
		%	(4.5,0) |- (3.5,-0.5) |- (2.5,-1.5) |- (1.5,-2.5) |- (.5,-3.5) |-
		%	(0,-4.5) -| (-0.5,-3.5) -| (-1.5,-2.5) -| (-2.5,-1.5) -| (-3.5,-.5) -|
		%	(-4.5,0) |- (-3.5,.5) |- (-2.5,1.5) |- (-1.5,2.5) |- (-.5,3.5) |- (0,4.5);%diameter. cubistic
		\draw[dashed,blue]
			(0,4.55) -| (0.55,3.55) -| (1.55,2.55) -| (2.55,1.55) -| (3.55,.55) -|
			(4.55,0) |- (3.55,-0.55) |- (2.55,-1.55) |- (1.55,-2.55) |- (.55,-3.55) |-
			(0,-4.55) -| (-0.55,-3.55) -| (-1.55,-2.55) -| (-2.55,-1.55) -| (-3.55,-.55) -|
			(-4.55,0) |- (-3.55,.55) |- (-2.55,1.55) |- (-1.55,2.55) |- (-.55,3.55) |- (0,4.55);%diameter. cubistic
		\node at (0,-5) {c)};
	\end{tikzpicture}
    \hskip 2ex
    % The graph
	%\includegraphics[height=6.4cm]{4+4i.pdf}
	%\includegraphics[height=7.2cm]{4+4i_thin.pdf}
	%\raisebox{.4cm}{\includegraphics[height=6.4cm]{4+4i_thin.pdf}}
	%\raisebox{.8cm}{\includegraphics[height=5.6cm]{4+4i_thin.pdf}}
	\begin{tikzpicture}
		\node {\includegraphics[height=5.6cm]{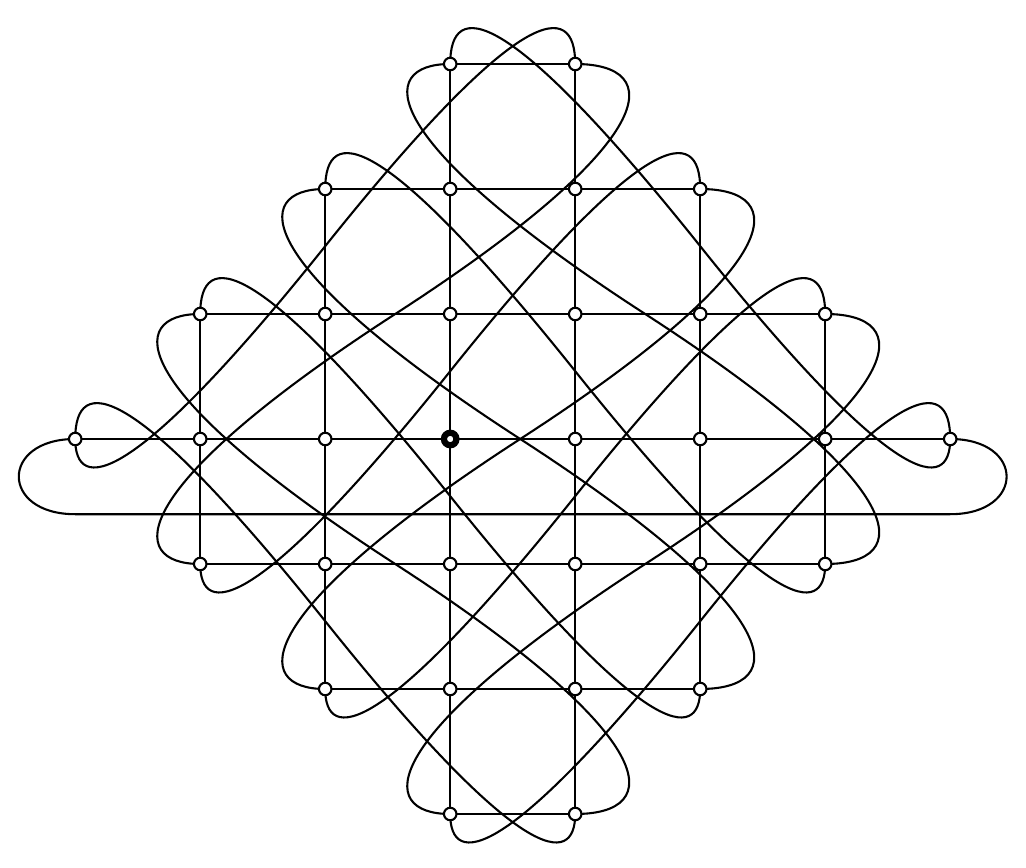}};
		\node at (0,-4) {d)};
	\end{tikzpicture}
	\typeout{Drawn Figure}
	\end{center}
	\caption{a) A 3-correcting and 4-covering linear Lee code over $\mathbb Z_{16}^2$. b) A Voronoi tessellation of the code. c) The tile associated to a codeword. d) The associated graph $Cay(\frac{\mathbb{Z}_{16}^2 }{\mathcal{C}};\{\pm e_1,\pm e_2\})$ in minimum distance representation.}
	\label{fig:code-tessellation-graph}
\end{figure*}
\end{remark}

The remainder of the paper describes a family of Cayley graphs over Gaussian integers. Let us denote by $\mathbb{Z}[i] = \{ a+bi \ | \ a, b \in \mathbb{Z}\}$ the ring of the Gaussian integers. In \cite{Hardy} the fundamentals on this ring can be found. Given an integer prime $p$, let us denote by $\mathbb{Z}[i]/p\mathbb{Z}[i]$ the quotient additive group of the Gaussian integers over the group generated by $(p) \subset \mathbb{Z}[i]$. Thus, the graph is defined as follows.

\begin{definition} Given an integer prime $p$, let us define the Cayley graph $\mathcal{G}_p = Cay(\mathbb{Z}[i]/p\mathbb{Z}[i], H),$  where
		$$H = \{\beta \in \mathbb{Z}[i]/p\mathbb{Z}[i] \mid \mathcal N(\beta)=1\}.$$
\end{definition}

Note that in the previous definition $\mathcal{N}(\beta) = \mathcal{N}(b_1 + b_2 i) = b_1 ^2 + b_2 ^2$ denotes the norm of $\beta$. Moreover, the adjacency in the graph is determined by the elements with unitary norm. In the subsequent sections, it will be proved that $\mathcal{G}_p$ induces a 2-quasi-perfect Lee code over $\mathbb{Z}_p^n$ under some conditions. Therefore, it must be determined which primes $p$ are such that $\mathcal{G}_p$ has error correction capacity 2 and diameter 3, as proved in Theorem~\ref{theo:grafo-codigo}.

\section{Error Correction Capacity of $\mathcal{G}_p$}\label{sec:correcion}

As explained in previous section, 2-quasi-perfect Lee codes will be obtained by means of Cayley graphs. In particular, it will be determined under which conditions the Cayley graph $\mathcal{G}_p$ over the additive group $\mathbb Z[i]/p\mathbb Z[i]$ and generating set the elements with unitary norm, induces a 2-quasi-perfect code. In this section it will be proved that $p \equiv \pm 5\pmod {12}$ implies that $\mathcal{G}_p$ has error correction capacity 2 over $\mathbb{Z}_{p}^n$ for $n = 2[\frac{p}{4}]$. Hence, in the remainder of the paper, let us assume that $p > 2$ is a prime integer. Therefore, the natural number $n=2[\frac{p}{4}]$ fulfills $p = 2n \pm 1$.

First, let us introduce some notation. Given a Gaussian integer $\beta = b_1 + b_2 i \in \mathbb{Z}[i]$, $\beta^{*}$ will denote its conjugate, that is $\beta^{*} = b_1 - b_2 i$. Also, $\Re(\beta) = b_1$ will stand for its real part and $\Im(\beta) = b_2$ for its imaginary part. Then, the following formula about the norm of a sum of Gaussian integers will be useful in several points of this paper.

\begin{lemma}\label{lem:normsum} For any pair of Gaussian integers $\beta,\gamma\in\mathbb Z[i]$,
	$$\mathcal N(\beta+\gamma)=\mathcal N(\beta)+\mathcal N(\gamma)+2\Re(\beta\gamma^*).$$
\end{lemma}

Then, the previous result can be used to prove the following technical lemma:

\begin{lemma}\label{lem:samenorms} For any $\gamma_1,\gamma_2\in\mathbb Z[i]/p\mathbb Z[i]$, if $\mathcal N(\gamma_1)=\mathcal N(\gamma_2)$ and $\mathcal N(1+\gamma_1)=\mathcal N(1+\gamma_2)$ then $\gamma_1\in\{\gamma_2,\gamma_2^*\}$.
\end{lemma}

\begin{proof} Since $\mathcal N(1+\gamma_1)=\mathcal N(1+\gamma_2)$, by Lemma~\ref{lem:normsum} it is obtained that $\Re(\gamma_1)=\Re(\gamma_2)$.
Therefore, there are $x,y,z\in\mathbb Z/p\mathbb Z$ such that $\gamma_1=x+yi$ and $\gamma_2=x+zi$.
Now, $\mathcal N(\gamma_1)=\mathcal N(\gamma_2)$ implies that $x^2+y^2=x^2+z^2$. As a consequence, $y^2=z^2$ and therefore $y\in\{\pm z\},$ which means $\gamma_1\in\{\gamma_2,\gamma_2^*\}$.
\end{proof}

\begin{corollary}\label{cor:zerodivisors}
Let $\beta \in \mathbb{Z}[i]/p\mathbb Z[i]$ be such that $\mathcal N(\beta)=1$. Then, $1+\beta$ is not a proper zero divisor.
\end{corollary}

\begin{proof}If $1+\beta$ is a zero divisor then $\mathcal N(1+\beta)=0=\mathcal N(1+(-1))$.
By Lemma~\ref{lem:samenorms}, $\beta\in\{-1,-1^*\}=\{-1\}$ and $1+\beta=0$.
\end{proof}

Let us denote by $G = \mathcal{U}(\mathbb{Z}[i]/ p \mathbb{Z}[i])$ the multiplicative group formed by the units of the ring. Then,
the set $$ H = \{\beta \in G  \mid \mathcal N(\beta)=1\}$$ is clearly a multiplicative normal subgroup of $G$.
It is actually a cyclic group, although this fact will not be used in the proofs.
Note that $H$ is the set of adjacencies of $\mathcal G_p$, that is, $\mathcal G_p=Cay(\mathbb{Z}[i]/ p \mathbb{Z}[i],H)$.
For any $\gamma\in \mathbb{Z}[i]/ p \mathbb{Z}[i]$, the following notation is introduced:
	$$\gamma H =\{\gamma\beta \mid \beta\in H \}.$$
Notice that if $\gamma\in G$, then $\gamma H$ is the \textsl{coset} of $H$ in $G$ with respect to $\gamma$.
Nevertheless, this notation is also defined for elements outside $G$, \textsl{i.e.}, for zero divisors of $\mathbb{Z}[i]/ p \mathbb{Z}[i]$.

The following lemma tells us that cosets can be identified by the norms of its elements.

\begin{lemma}\label{lem:cosetnorm}For any $\gamma \in G$, $$\gamma H=\{\beta\in\mathbb Z[i]/p\mathbb Z[i]\mid \mathcal N(\beta)=\mathcal N(\gamma)\}.$$
\end{lemma}

\begin{proof} In order to prove the sets equality, it will be first proved that $\gamma H \subseteq \{\beta\in G \mid \mathcal N(\beta)=\mathcal N(\gamma)\}.$ Thus, let us consider $\beta\in \gamma H$ and it has to be proved that $\mathcal N(\beta)=\mathcal N(\gamma)$. Since $\beta\in \gamma H,$ then there exists $\eta\in H$ such that $\beta=\gamma\eta$. Hence $\mathcal N(\beta)=\mathcal N(\gamma)\mathcal N(\eta)=\mathcal N(\gamma)$.

Now, let us consider the other inclusion, that is, $\gamma H \supseteq \{\beta\in G \mid \mathcal N(\beta)=\mathcal N(\gamma)\}$. Therefore, let $\beta \in G$ be such that $\mathcal N(\beta)=\mathcal N(\gamma)$.
Since $\gamma$ is invertible, $\beta=\gamma(\beta\gamma^{-1})$.
Now, as $\mathcal N(\beta\gamma^{-1})=1$ it is obtained that $\beta\in\gamma H$.
\end{proof}

Theorem~\ref{thm:cardinal} states that the degree of the graph $\mathcal{G}_p$ is $2n$. To prove it some particular cases of the Quadratic Reciprocity Law will be necessary, which are recalled in the following theorem for self-containedness.

\begin{theorem}[Quadratic Reciprocity]\label{theo:qrl}
If $p$ is an integer prime, then:
\begin{enumerate}
\item The number of solutions to $-1=x^2$ in $\mathbb Z/p\mathbb Z$
	is:
\begin{itemize}
\item 2 if $p\equiv 1\pmod 4$,
\item 1 if $p=2$ and
\item 0 if $p\equiv 3\pmod 4$.
\end{itemize}
\item The number of solutions to $3=x^2$ in $\mathbb Z/p\mathbb Z$
is:
\begin{itemize}
\item 2 if $p\equiv \pm 1\pmod {12}$,
\item 1 if $p=3$ or $p=2$ and
\item 0 otherwise.
\end{itemize}
\end{enumerate}

\end{theorem}

\begin{theorem}\label{thm:cardinal}
For any odd prime integer $p$, let $n=2[\frac{p}{4}]$.  Then,
$$|H|= |\{\beta\in\mathbb Z[i]/p\mathbb Z[i]\mid \mathcal N(\beta)=1\}| = 2n.$$
\end{theorem}

\begin{proof} It is clear that
$$|H|=|\{(x,y)\mid x,y\in\mathbb{Z}/p\mathbb{Z},\ x^2+y^2=1\}|.$$

Therefore, let us consider the solutions of $x, y\in \mathbb{Z}/p\mathbb{Z}$ of equation $x^2+y^2=1$. First, if $x=1$ then $y^2=0$ whose unique solution is $y=0$.
Let us assume $x\neq 1$ to find the rest of the solutions. Since $x\neq 1$, $x-1$ has inverse and it is possible to define $s=y/(x-1) \in \mathbb{Z}/p\mathbb{Z}$.
By considering the intersection of the straight line $y=s(x-1)$ with the curve $x^2+y^2=1$ it is obtained that $x^2+(s(x-1))^2=1$. The only solutions of this equation are $x=1$ (which has already been considered)
and $x=\frac{s^2-1}{s^2+1}$. This second solution for $x$ equals 1 if and only if $p=2$. Thus, the only solutions with $x\neq 1$ are $x=\frac{s^2-1}{s^2+1}$ and $y=\frac{-2s}{s^2+1}$.

Now, for each possible value of $s$, there is one solution with this form, that is, $p$ minus the number of solutions of $s^2+1=0$. By the Quadratic Reciprocity Law (first item of Theorem~\ref{theo:qrl}) there are $p+1$ solutions if $p\equiv 3 \pmod 4$ and $p-1$ if $p\equiv 1 \pmod 4$. Thus, for primes of the form $p=1+4k$, there are $p-1=4k=2n$ solutions and for primes $p=-1+4k$ there are $p+1=4k=2n$ solutions, where $k \in \mathbb{N}$.

Finally, just to ensure that the counted solutions are all different, note that if for a pair $s_1, s_2$ the same solution $(x,y)$ is obtained, then $s_1=s_2=y/(x-1)$.
\end{proof}

Next, it is easy to obtain the following consequence of previous theorem, which will be used in Section~\ref{sec:diametro} to determine the diameter of the graph $\mathcal{G}_p$.

\begin{corollary}\label{cor:cardinalcoset}
For any odd prime integer $p$, let $n=2[\frac{p}{4}]$. If $0 \neq \gamma \in \mathbb{Z}[i]/p\mathbb{Z}[i]$ then $|\gamma H|=2n.$
\end{corollary}

\begin{proof} Firstly, note that if $\gamma\in G$, then $\gamma H$ is a coset, which is widely known to have the same cardinal. Thus, the non-immediate part of the proof lies on the zero divisors.
By Theorem~\ref{thm:cardinal}, it is straightforward that $|\gamma H|\leq 2n$.
Proceeding by \textsl{reductio ad absurdum}, let us assume $|\gamma H| < 2n$. Then, there exist $\beta_1\neq \beta_2$ such that $\gamma\beta_1=\gamma\beta_2$, thus $\gamma(\beta_1-\beta_2)=0$.
Since $\gamma\neq 0$ then $\beta_1-\beta_2$ must be a zero divisor. Now, multiplying by $\beta_1^{-1}$, $1-\beta_2\beta_1^{-1}$ is also a zero divisor.
By Corollary~\ref{cor:zerodivisors}, $1-\beta_2\beta_1^{-1}=0$ and hence $\beta_1=\beta_2$, which is a contradiction.
\end{proof}

Before stating the conditions under which $\mathcal{G}_p$ has error correction capacity 2, the following lemma is going to be proved.
This lemma determines the number of possible norms among the neighbours of a vertex with a given norm.

\begin{lemma}\label{lem:degree} For any $c\in\mathbb Z/p\mathbb Z$, $c\neq 0$, let us consider the set
$N_p(c) = \{\mathcal N(1+\beta)\mid \mathcal N(\beta)=c\} \subset \mathbb{Z}/p\mathbb{Z}$. Then, it is obtained that:
	$$
	|N_p(c)|
	=\begin{cases}
	n+1	&	\text{ if $c$ is a square residue \textrm{mod $p$}},\\
	n	&	\text{ if $c$ is not a square residue \textrm{mod $p$}.}
	\end{cases}
	$$
\end{lemma}

\begin{proof} In the first case, that is $c$ being a square residue, there must exists $s\in\mathbb Z/p\mathbb Z$ such that $c=s^2$.
By Lemma~\ref{lem:cosetnorm} and Corollary~\ref{cor:cardinalcoset} there are $2n$ elements with norm $c$, which are:
\begin{multline*}
	\{\beta\mid \mathcal N(\beta)=c\}\\
	=\{s,-s,\beta_1,\beta_2,\dotsc,\beta_{n-1},\beta_1^*,\beta_2^*,\dotsc,\beta_{n-1}^*\},
\end{multline*}for some $\beta_1,\dotsc,\beta_{n-1}\in \mathbb Z[i]/p\mathbb Z[i]$.
Then,
\begin{multline*}
		N_p(c) = \{\mathcal N(1+\beta)\mid \mathcal N(\beta)=c\}=\\
	\{\mathcal N(1+s),\mathcal N(1-s),\mathcal N(1+\beta_1),\mathcal N(1+\beta_2),\dotsc,\mathcal N(1+\beta_{n-1})\},
\end{multline*}which are different by Lemma~\ref{lem:samenorms}.
Hence $|N_p(c)|=2+(n-1)=n+1$.

For the case of $c$ being a square non-residue let us proceed in a similar way. It is obtained that
\begin{multline*}
	\{\beta\mid \mathcal N(\beta)=c\}\\
	=\{\beta_0,\beta_1,\beta_2,\dotsc,\beta_{n-1},\beta_0^*,\beta_1^*,\beta_2^*,\dotsc,\beta_{n-1}^*\}.
\end{multline*}
Then
\begin{multline*}
		N_p(c)=\{\mathcal N(1+\beta)\mid \mathcal N(\beta)=c\}=\\
	\{\mathcal N(1+\beta_0),\mathcal N(1+\beta_1),\mathcal N(1+\beta_2),\dotsc,\mathcal N(1+\beta_{n-1})\},
\end{multline*}which are different by Lemma~\ref{lem:samenorms}.
Hence $|N_p(c)|=n$.
\end{proof}

As it is noted afterwards, the case $c=1$ in previous lemma will be used to prove the error correction capacity. Later,
the fact that $n$ is a lower bound of $|N_p(c)|$ will be considered to determine the graph diameter.

To finish the section, next theorem establishes the conditions for $p$ such that $\mathcal{G}_p$ has error correction capacity 2.

\begin{theorem}\label{thm:correction} Let $p$ be a prime integer satisfying $p\equiv \pm 5 \pmod{12}$.
Then, the Cayley graph $\mathcal{G}_p$ has error correction capacity 2.
\end{theorem}

\begin{proof}
Let $n=2\left[\frac{p}{4}\right]$.
As it was explained in previous section, it has to be proved that $\mathcal{G}_p$ contains $|B_2^n|=2n^2+2n+1$ vertices at distance 2 or less from 0. Clearly, 0 is the unique vertex at distance 0. Now, the set $H$ contains all the vertices at distance 1 and $|H|=2n$ by Theorem~\ref{thm:cardinal}.

The vertices at distance 2 is the set $A=\{\beta_a+\beta_b\mid \beta_a,\beta_b\in H\}\setminus (H \cup \{0\})$. Thus, let us prove that $|A|=2n^2$.
By Lemma~\ref{lem:cosetnorm} and Corollary~\ref{cor:cardinalcoset},
$|A|=2n\cdot |N_p(1)\setminus \{0,1\}|$.
Since 1 is always a square residue for any $p$, hence by Lemma~\ref{lem:degree},
$|N_p(1)\setminus \{0\}| = n$.
It remains to be proved that 1 does not belong to $N_p(1)$.

Suppose that there is $\beta$ with $\mathcal N(\beta)=1$ and $\mathcal N(1+\beta)=1$. Then, by Lemma~\ref{lem:normsum}, $1=2+2\Re(\beta)$ and hence $\Re(\beta)=-2^{-1}$.
Let $\beta=-2^{-1}+yi$, which implies $1=\mathcal N(\beta)=2^{-2}+y^2$.
Then, $3=(2y)^2$, which only has solutions for $p=3$ or $p\equiv \pm 1\pmod {12}$ by
the second item of Theorem~\ref{theo:qrl}.
%Theorem~\ref{theo:qrl}, item X.
Thus, $|N_p(1)\setminus \{0,1\}|=|N_p(1)\setminus \{0\}|=n$ and $|A|=2n\cdot n$, which concludes the proof.
\end{proof}

\begin{remark} If $p$ is a prime greater than $3$ that does not satisfy $p\equiv \pm 5\pmod {12}$, then $p\equiv \pm 1\pmod {12}$.
In this case, $\mathcal G_p$ only contains $2n^2+1$ vertices at distance 2 or less from vertex 0. Although it is not a 2-error correcting code, it is very close to it, since only
$2n$ syndromes cannot be corrected.
\end{remark}

\section{Diameter of $\mathcal{G}_p$}\label{sec:diametro}

In this section it will be proved that $\mathcal{G}_p$ has diameter 3 for any prime $p > 5$. The proof will be divided into two subsections. The first considers the case $p\equiv 3 \pmod 4$ and the second the case $p\equiv 1 \pmod 4$. Also, from here onwards it will be assumed again that $n = 2 [\frac{p}{4}]$.
Note that, since $|\mathbb{Z}[i]/ p \mathbb{Z}[i]|=p^2>|B_2^n|$, there are vertices outside the sphere of radius 2, which means that the diameter of the graph is at least 3.
As it will be seen next, the proofs proceed by \textsl{reductio ad absurdum} by the assumption of the existence of a vertex at a distance 4 from vertex 0, thus reaching a contradiction.

\subsection{Case $p\equiv 3 \pmod 4$}

In this case the proof of the diameter can be easily obtained by using a counting argument.
Note that in this case $p=2n-1$ and therefore $\mathbb{Z}[i]/ p \mathbb{Z}[i]$ is a field.

\begin{theorem} For any prime $p$ such that $p\equiv 3\pmod 4$ the graph $\mathcal G_p$ has diameter 3.
\end{theorem}

\begin{proof}
By \textsl{reductio ad absurdum} let us assume that there exists a vertex $\gamma\in\mathbb{Z}[i]/ p \mathbb{Z}[i]$ at distance 4 from vertex 0. Let $c=\mathcal N(\gamma)$. Since $\gamma$ is so far, it is obtained that $N_p(1)\cap N_p(c)=\emptyset$.

Let us denote by $W_t(0)$ the number of vertices at a distance $t$ from vertex 0. Then, $\{W_t(0) \mid t = 0, \ldots, 4\}$ is the distance distribution of the graph $\mathcal{G}_p$. Now, the cardinals $W_1(0) = |H|$ and $W_4(0) \geq |\gamma H|$ can be calculated by Corollary~\ref{cor:cardinalcoset}. Also, by  Lemma~\ref{lem:degree} it can be computed that $|N_p(1)|=n+1$ and $|N_p(c)|\geq n$. Thus, the obtained bounds for the distance distribution are summarized as follows:  $$\begin{array}{ll}
W_0(0) = |\{0\}|						         &	=1\\
W_1(0) = |H|							         &	=1\cdot 2n\\
W_2(0) = 2n\cdot |N_p(1)\setminus \{0,1\}|	     &	\geq (n-1)\cdot 2n\\
W_3(0)	\geq 2n\cdot |N_p(c)\setminus \{c\}|	 &	\geq (n-1)\cdot 2n\\
W_4(0)	\geq |\gamma H|						     &	=1\cdot 2n
\end{array}$$
As a consequence, the total number of vertices satisfies $|\mathbb{Z}[i]/ p \mathbb{Z}[i]|\geq 1+2n(1+(n-1)+(n-1)+1)=4n^2+1>4n^2-4n+1=p^2=|\mathbb{Z}[i]/ p \mathbb{Z}[i]|$, which is a contradiction.
\end{proof}

\subsection{Case $p\equiv 1 \pmod 4$}

Unfortunately, the reasoning made in the previous case fails to give us a contradiction if $p\equiv 1\pmod 4$. Therefore, it will be needed to resort to the tight bound from algebraic geometry obtained in the Hasse--Weil Theorem. Note that, in this case, $p=2n+1$ and the ring $\mathbb{Z}[i]/ p \mathbb{Z}[i]$ contains zero divisors.\par

First, let us prove two technical lemmas that analyze what happens with the zero divisors of the ring.

\begin{lemma}\label{lem:zH} For any proper zero divisor $\zeta\in\mathbb Z[i]/p\mathbb Z[i]$,
	$$
	\zeta H=\{x\zeta\mid x\in\mathbb Z/p\mathbb Z,\ x\neq 0\}.
	$$
\end{lemma}

\begin{proof} On the one hand, by Corollary~\ref{cor:cardinalcoset}, the cardinal $|\zeta H| $ is $2n$. On the other hand, $|\{x\zeta\mid x\in\mathbb Z/p\mathbb Z,\ x\neq 0\}|$ has $p-1=2n$ elements. Since both sets have the same size, it is enough to prove one inclusion to show the sets equality. Therefore, let us prove
the left to right inclusion.

Let $\beta=a+bi$ be an element of norm 1 and $\zeta=u+vi$ a proper zero divisor, hence of norm 0.
As $\zeta\neq0$ and $\mathbb Z/p\mathbb Z$ is a field, both $u$ and $v$ are nonzero.
Let us define $x=a-b\frac{v}{u}\in \mathbb Z/p\mathbb Z$. Therefore,
\begin{multline*}
x\zeta=(a-b\frac{v}{u})(u+vi)
=(au-bv)+(av-b\frac{v^2}{u})i\\
=(au-bv)+(av-b\frac{-u^2}{u})i
=(au-bv)+(av+bu)i\\
=(a+bi)(u+vi)
=\beta\zeta.
\end{multline*}
Finally, note that if $x$ were zero, then $\beta$ would be a zero divisor, contradicting $\mathcal N(\beta)=1$.
\end{proof}

The following lemma has its inspiration in Lemma~\ref{lem:degree}, but with the intention of generalizing to the case of zero divisors.

\begin{lemma}\label{lem:neigh0}For any proper zero divisor $\zeta\in\mathbb Z[i]/p\mathbb Z[i]$,
	$$
	\{\mathcal N(\beta+\zeta)\mid \mathcal N(\beta)=1\}=
		\mathbb Z/p\mathbb Z\setminus \{1\}.
	$$
\end{lemma}

\begin{proof} Let $\zeta=u+vi$ be a proper zero divisor. By Lemma~\ref{lem:zH},
\begin{align*}
\{\mathcal N(\beta+\zeta)\mid \mathcal N(\beta)=1\}
&=\{\mathcal N(1+\beta\zeta)\mid \mathcal N(\beta)=1\}\\
&=\{\mathcal N(1+x\zeta)\mid x\in\mathbb Z/p\mathbb Z,\ x\neq 0\}\\
&=\{1+2xu\mid x\in\mathbb Z/p\mathbb Z,\ x\neq 0\}.
\end{align*}

To finish, note that $y=1+2xu$ with $x\neq 0$ has solution for every value of $y$ except 1.
\end{proof}

The previous lemma indicates that proper zero divisors are neighbours of every vertex at distance 2 from 0, and hence they are at distance 3 from 0. Then, the following lemma gives a polynomial
description of the sets $N_p(t)$.

\begin{lemma}\label{lem:polynomial}Let $p\equiv 1\pmod 4$ be a prime in $\mathbb Z$. For any $t\in\mathbb Z/p\mathbb Z$, $t\neq 0$, it is obtained that
	$$N_p(t)=\{x^{-1}(x+1)(x+t)\mid x\in \mathbb Z/p\mathbb Z,\ x\neq 0\}.$$
\end{lemma}

\begin{proof}
By the first item of Theorem~\ref{theo:qrl}, there exists $r\in\mathbb Z/p\mathbb Z$ such that $r^2=-1$.
Note that $x^{-1}(x+1)(x+t)=x+tx^{-1}+t+1$. First, let us prove the left to right inclusion of the sets.
In this aim, let $\beta=a+bi$, $\mathcal N(\beta)=a^2+b^2=t$ for a generic element $\mathcal N(1+\beta)$
in $N_p(t)$. Thus, let us check that $x=a+rb$ satisfies $\mathcal N(1+\beta)=x+tx^{-1}+t+1$.
By Lemma~\ref{lem:normsum}, $x\mathcal N(1+\beta)=x(\mathcal N(1)+\mathcal N(\beta)+2\Re(\beta))=x(t+1)+2ax$.
Hence,
\begin{align*}
	x&(x+tx^{-1}+t+1)-x\mathcal N(1+\beta)\\
	&=x^2+t-2ax\\
	&=t+(a+rb)^2-2a(a+rb)\\
	&=t+(a^2+2rab+r^2b^2)-(2a^2+2rab)\\
	&=t-a^2+r^2b^2\\
	&=t-a^2-b^2\\
	&=0
\end{align*}

For the right to left inclusion, let $x\neq 0$ and $y=x^{-1}(x+1)(x+t)$ being an element
of $\{x^{-1}(x+1)(x+t)\mid x\in \mathbb Z/p\mathbb Z,\ x\neq 0\}$.
Now, define $\beta=x+x^{-1}(t-x^2)+2^{-1}x^{-1}(t-x^2)ri$. Then, by calculation $\mathcal N(\beta)=(x+x^{-1}(t-x^2))^2+(2^{-1}x^{-1}(t-x^2)r)^2=t$.
Moreover, $\mathcal N(1+\beta)=1+t+2\Re(\beta)=1+t+2x+x^{-1}(t-x^2)=y$, which ends the proof.
\end{proof}

The intersection between $N_p(1)$ and $N_p(t)$ will be given by the roots of the polynomial $P_t(x,y)=y(x+1)^2-x(y+1)(y+t)$. In order to apply the Hasse--Weil bound, the polynomial must be irreducible. Therefore, let us introduce the  following definition and two useful results in Lemma~\ref{lemma:absolute} and Corollary~\ref{cor:irreducible}.

\begin{definition} Given a field $\mathbb F$, a polynomial $P\in \mathbb F[x,y]$ is called \textsl{absolutely irreducible} if it is irreducible in the algebraic closure of $\mathbb F$.
\end{definition}

\begin{lemma}\label{lemma:absolute} For any prime $p$, the polynomial $P_t(x,y)=y(x+1)^2-x(y+1)(y+t) \in \mathbb{Z}_p[x, y]$ is absolutely irreducible for $t\neq 0,1$.
\end{lemma}

\begin{proof}
The polynomial $P_t(x, y)=xy(x-y)+(1-t)xy+y-tx$ has degree 3. If $P_t(x, y)$ is not absolute irreducible, then there exist polynomials $A(x, y)$, $B(x, y)$ with coefficients in the algebraic closure of $\mathbb Z/p\mathbb Z$ such that $P_t(x, y)=AB$ with $\deg A(x, y)=2$ and $\deg B(x, y)=1$. Furthermore, the product of the leading terms of $A(x, y)$ and $B(x, y)$ must be $xy(x-y)$. Let us consider the following three mutually exclusive cases, depending on polynomials $A(x, y)$ and $B(x, y)$
\begin{enumerate}
\item Case $A(x, y)=(xy+ax+by+c)$, $B(x, y)=(x-y+d)$. The coefficient of $x^2$ in $A(x, y)\cdot B(x, y)$ is $a$ and the one of $y^2$ is $-b$. By hypothesis, both are 0 in $P_t(x, y)$. Then, the coefficient of $xy$ is $d=1-t$, the coefficient of $x$ is $c=-t$ and the coefficient of $y$ is $-c=t=1$. Hence, for $t=1$ there exists the factorization $P_1(x, y)=(xy-1)(x-y)$.
\item Case $A(x, y)=(x(x-y)+ax+by+c)$, $B(x, y)=(y+d)$. Now, the coefficient of $x^2$ in $A(x, y)\cdot B(x, y)$ is $d=0$ and the coefficient of $y^2$ is $b=0$. Then, the coefficient of $xy$ is $a=1-t$, the coefficient of $x$ is $0=-t$ and the coefficient of $y$ is $c=1$. Hence, for $t=0$ there exists the factorization $P_0(x, y)=(x^2-xy+x+1)y$.
\item Case $A(x, y)=(y(x-y)+ax+by+c)$, $B(x, y)=(x+d)$. The coefficient of $x^2$ is $a=0$ and the coefficient of $y^2$ is $-d=0$. Then, the coefficient of $y$ would be $0=1$, which implies that there exists no factorization.
\end{enumerate}
Finally, there are factorizations of $P_t(x, y)$ only for $t=0$ and $t=1$, which proves the result.
\end{proof}

\begin{corollary}\label{cor:irreducible} The homogeneous polynomial
	$$^hP_t(x,y,z)=xy(x-y)+(1-t)xyz+(y-tx)z^2$$ is absolutely irreducible for $t\neq 0,1$.
\end{corollary}
\begin{proof}
If $^hP_t(x,y,z)$ had a factorization, then its evaluation at $z=1$ would be a factorization of $P_t(x, y)$, contradicting Lemma~\ref{lemma:absolute}.
\end{proof}

Finally, let us conclude the section by proving the main result.

\begin{theorem}
If $p$ is a prime such that $p\equiv 1\pmod 4$ and $p>5$, then the diameter of $\mathcal G_p$ is 3.
\end{theorem}

\begin{proof} Let us proceed again by \textsl{reductio ad absurdum}. First, let us assume the existence of a vertex $\gamma$ at distance 4 from 0 in $\mathcal{G}_p$, with $p$ fulfilling the hypothesis of the statement. Let $t=\mathcal N(\gamma)$. Note that $t \neq 1$ since the vertices with norm equal to 1 are at distance 1. Also, $t \neq 0$ by Lemma~\ref{lem:neigh0}.
Hence, by Lemma~\ref{lem:cosetnorm}, the vertices with norm in the set $N_p(t)\setminus \{0\}$ are at distance at least 3.
Meanwhile, the vertices with norm in $N_p(1)\setminus \{0\}$ are at distance at most 2 from 0. Therefore, the intersection of previous two sets is $N_p(1)\cap N_p(t)=\{0\}$.

Now, using polynomial notation, previous sets equality is equivalent, by Lemma~\ref{lem:polynomial}, to the non-existence of solutions to $x^{-1}(x+1)^2=y^{-1}(y+1)(y+t)$ other than $x=-1$.
Let us highlight that the solution $x=-1$ corresponds with norm 0. Thus, vertices in $H$ have vertex 0 as their neighbour, while vertices in $\gamma H$ have as some of their neighbours vertices that are proper zero divisors.

The contradiction will be obtained when proving the existence of a solution to $P_t(x,y)=0$ other than the trivial ones $(x,y)\in\{(0,0),(-1,-1),(-1,-t)\}$. To this aim, let us define the varieties $$V_t=\{(x,y)\in (\mathbb Z/p\mathbb Z)^2\mid P_t(x,y)=0\},$$ $$X_t=\{(x:y:z)\in \mathbb P_{\mathbb Z/p\mathbb Z}^2\mid {}^hP_t(x,y,z)=0\},$$
where $\mathbb P_{\mathbb Z/p\mathbb Z}^2$ denotes the projective space of dimension 2 over $\mathbb Z/p\mathbb Z$. The notation $(x:y:z)$ indicates a projective point, which is the same point as $(\lambda x:\lambda y:\lambda z)$ for any $\lambda \neq 0$. Thus, affine solutions can be recovered by taking $\lambda =z^{-1}$; except for solutions $(x:y:0)$, which are the points at the infinite.

Hasse--Weil's theorem~\cite{Cohen} states that $$\bigl||X_t|-(p+1)\bigr|\leq 2\sqrt{p},$$
for absolutely irreducible polynomial curves $X_t$ of degree 3. Note that, by Corollary~\ref{cor:irreducible}, Hasse--Weil's theorem can be applied to $^hP_t(x,y,z)$. Therefore,
$$
	|X_t| \geq p+1-2\sqrt{p}.
$$

Now, the only 3 projective solutions for ${}^hP_t(x,y,z)=0$ with $z=0$ are $(x:y:z)\in\{(0:1:0),(1:0:0),(1:1:0)\}$. Thus, $|V_t| = |X_t| - 3$, which implies:
	$$
	|V_t| \geq p-2-2\sqrt{p}.
	$$
As a consequence, those primes $p$ such that $|V_t|\geq 4$ provide the expected contradiction. Clearly, if $p\geq 17$ then,
	$$|V_t|\geq p-2-2\sqrt{p} \geq 17 -2-2\sqrt{17}\geq 6.7.$$

Finally, the unique prime $p\equiv 1 \pmod 4$ such that $5 < p < 17$ is 13. In this particular case, it can be computed that $|V_t|\geq 9$ for any $t$, which concludes the proof.
\end{proof}

\begin{remark} $\mathcal{G}_5$ has diameter 4 since vertex $2+2i$ and its associates are at distance 4 from vertex 0.
\end{remark}

\section{Discussion}\label{sec:conclusiones}

In this final section, conclusions of this work and future research will be presented. In the first subsection, the main result is rewritten using parity-check matrices. Besides, a formal proof of the infiniteness of the constructed family of quasi-perfect codes is given. Some considerations on the density of the codes are taken into account. Moreover, other examples of codes presenting greater density and an upper error correction capacity are shown. In the final subsection, the authors exhibit the relations between the graphs considered in the present study with other graph theoretical problems, trying to give a new insight into the perfect Lee codes conjecture formulated by Golomb and Welch more than forty years ago.

\subsection{Quasi-perfect Lee codes}

%%% CARMEN: Hay que dejar claros los conceptos dimension, densidad de codigo y el espacio donde tenemos el codigo (multiplo)

As it has been proved in previous Sections \ref{sec:correcion} and \ref{sec:diametro}, $\mathcal{G}_p$ has error correction capacity 2 and diameter 3, for any prime $p > 5$ and $p \equiv \pm 5 \pmod{12}$.

%Dirichlet's theorem on arithmetic progressions asserts that in an arithmetic progression there are infinite primes.
Dirichlet's theorem on arithmetic progressions asserts that, in any arithmetic progression whose initial term is coprime with its increment, there are infinitely many primes.
As a natural consequence, congruences can be considered as arithmetic progressions, and therefore it can be obtained:

\begin{corollary} There are infinitely many $n \in \mathbb{N}$ such that $p = 2n \pm 1$, $p \geq 7$ prime in $\mathbb{Z}$, $p \equiv \pm 5 \pmod{12}$.
\end{corollary}

Then, when applying the previous result it is obtained:

\begin{corollary}
	%The family of graphs $\{ \mathcal{G}_p \ | \ p\text{ prime }p \equiv \pm 5 \pmod{12} \}$ is infinite.
	The family of graphs $\mathcal{G}_p$ contains infinitely many graphs with error correction capacity 2 and diameter 3.
\end{corollary}

Then, as it was discussed in Section~\ref{sec:grafo-codigo}, each of these graphs induces a 2-quasi-perfect Lee code.

\begin{theorem} \label{theo:paritycheck} Let $p$ be a prime. Let $\{ \beta_1 ,  \ldots , \beta_{2n} \} = \{\pm \beta_1 , \ldots  , \pm \beta_{n} \}$ be the elements
of $\mathbb{Z}[i]/p\mathbb{Z}[i]$ with unitary norm. Let $\mathcal{G}_p = Cay (\mathbb{Z}[i]/p\mathbb{Z}[i], \{ \beta_1 ,  \ldots , \beta_{2n} \})$ and $\mathcal C \subset \mathbb{Z}_{p}^n$ be the code associated to $\mathcal{G}_p$ such that $\mathcal C = \ker \phi$ and $\phi: \mathbb{Z}_{p}^n \longrightarrow \mathbb{Z}[i]/p\mathbb{Z}[i]$ such that $\phi(e_j) = \beta_j$. Then,

$$ M = \begin{pmatrix}
	\Re(\beta_1) & \Re(\beta_2) &   \cdots & \Re(\beta_{n})   \\
	\Im(\beta_1) & \Im(\beta_2) &   \cdots & \Im(\beta_{n}) \\
\end{pmatrix}.
$$

is the parity-check matrix of $\mathcal C$.
\end{theorem}

\begin{proof} Let us denote by $\psi: \mathbb{Z}[i] \longrightarrow \mathbb{Z}^2$ such that $\psi(\beta) = (\Re(\beta) ,  \Im(\beta)).$ Note that the homomorphism defined by the matrix $M$ is equal to the mapping $\psi \circ \phi$. Thus, $\mathcal C = \ker(\phi) = \ker(\psi \circ \phi) = \{x \mid Mx= 0\}.$
\end{proof}

Now, let us give some considerations on the quality of the constructed codes. Note that, since the Lee sphere of radius 2 contains $|B_2|=2n^2+2n+1$ words, the graph induced by any 2-quasi-perfect linear code has at least $2n^2+2n+1$ vertices.
The graphs $\mathcal{G}_p$ constructed in this paper have $p^2$ vertices. Therefore, for the case $p=2n+1$, the number of vertices is $p^2=4n^2+4n+1=2|B_2|-1$. Also, for the case $p=2n-1$, the number of vertices is $p^2=4n^2-4n+1=2|B_2|-8n-1$. Thus, the reached vertices are asymptotically the double of those that would be reached in the graph associated to a perfect code. In other words, the density of the codes presented is $\frac{1}{p^2}$.

%like our construction
%For $r=2$ $n=4$, $p=7$, $S=\{1,2+2i\}$.

% Not a lot better
%For $r=2$ $n=6$, $p=12$, $S=\{1,1+3i,4+6i\}$.

Although the obtained density is quite good, for some small cases (low dimensions), graphs with a smaller number of vertices have been computationally found.
Let us consider the following examples.

\begin{table*}
\begin{center}
\begin{tabular}{|c|c|c|c|c|}
  \hline
  $n$ & $p$ & $H$ & $p^2$ & $|B_3^n|$  \\
  \hline
  4 & 13 & $\pm\{1,3+4i,  i,-4+3i\}$ & 169 & 129  \\
  6 & 26 & $\pm\{1,4+4i,9+11i,  i,-4+4i,-11+9i\}$ & 676 & 377  \\
  8 & 41 & $\pm\{1,2+13i,6+18i,11+i,  i,-13+2i,-18+6i,-1+11i\}$ & 1681 & 833  \\
  \hline
\end{tabular}
\caption{$\mathcal G_p$ graphs that generate 3-quasi-perfect Lee codes over $\mathbb{Z}_p^n$.}\label{tabla:codigos}
\end{center}
\end{table*}

\begin{example} Let $n=8$ be the dimension and $p=13$. The set of generators of the Cayley graph will be
%$H=\pm\{1,4+10i,8,7+11i\}$.
$H=\bigl\{\beta u\mid \beta \in\{1,4+10i,8,7+11i\},\ u\in \mathcal U(\mathbb Z[i])\bigr\}$.
In this case the Cayley graph
$\mathcal{G} = Cay(\mathbb{Z}[i]/p\mathbb{Z}[i], H)$ induces a 2-quasi-perfect code. Note that $\mathcal{G}$ has $p^2 = 169$ vertices, which is just 17\% over $|B_2^8|=145$, the cardinal of the sphere in this dimension.
\end{example}

\begin{example} Let $n=16$ be the dimension. In this case, by extending the search into a different ring, a new graph has been found. The graph is built over the Quaternion integers $\mathbb H(\mathbb Z)$ modulo $p=5$, being the generator set
%$H=\pm\{1,\allowbreak 1+2i+3j,\allowbreak 3i+4j+1k,\allowbreak 3+4i+3j\}$.
$H=\bigl\{\beta u\mid \beta\in\{1,\allowbreak 1+2i+3j,\allowbreak 3i+4j+4k,\allowbreak 3+4i+3j\},\ u\in \mathcal U(\mathbb H(\mathbb Z))\bigr\}$.
In this case, the number of vertices of the graph is $p^4=625$, which is 15\% over $|B_2^{16}|=545$.
\end{example}

These small examples suggest that there exist codes very close to be perfect, although general constructions seem to be difficult to find.

%From the result obtained by Post~\cite{Post} it was deduced that there is no periodic perfect code with radius greater than the dimension of the space.
%Previously to that paper, Golomb and Welch~\cite{Golomb} had already noted that there cannot be perfect codes with correction greater than a constant that depends on the dimension by the use of the maximum density of packing with cross-polytopes.
Golomb and Welch in \cite{Golomb} noted that there cannot be perfect codes with correction greater than a constant that depends on the dimension by the use of the maximum density of packing with cross-polytopes.
Clearly, this can be applied to quasi-perfect codes. For every $n$ there exists $t_n$ such that there are no $t$-quasi-perfect codes for $t\geq t_n$. Hence, this might suggest that the radius 2 case is an exceptional one.
Nevertheless, some 3-quasi-perfect codes have been found for small dimensions. Note that in this case the $n$-dimensional sphere of radius 3 has cardinal $|B_3^n|=\frac{1}{3}(1+2n)(3+2n+2n^2)$. The examples that we have found are summarized in Table~\ref{tabla:codigos}. The codes are obtained from Cayley graphs $Cay(\mathbb{Z}[i]/p\mathbb{Z}[i], H)$, for parameters $n, p, H$ as shown in the table. As it can be seen, the first example is just 31\% over the cardinal of the sphere, while the second and third are 79\% and 102\%, respectively. Any of the three examples can be considered as 3-quasi-perfect codes really near to the perfect code.

In the authors' opinion, the construction of an infinite family of graphs containing these codes or similar ones would have a great practical value. Moreover, it would contribute to a better understanding of the Golomb and Welch conjecture.

\subsection{Related Problems}

This study could be used to deal with problems from areas of study different from Coding Theory. For example, this graph theoretical study of perfect codes can be seen as the reverse of the degree-diameter problem for Cayley graphs over Abelian finite groups~\cite{Miller}. In this problem, for a given diameter, graphs with the maximum possible number of vertices are searched. Specifically, for a positive integer $t$, graphs providing $t$-covering codes but without considering the correction are looked for. Note that in this case, the order of the graphs obtained is lower than the cardinal of the corresponding sphere $|B_t^n|$. Therefore, in the present paper graphs providing $t$-correcting codes and enforcing additionally $(t+1)$-covering have been constructed. In our case, the order of the Cayley graphs is always greater than the cardinal of the sphere $|B_t^n|$. The degree-diameter problem
for $t=2$ and $t=3$ has been considered in \cite{Macbeth,Vetrik}.
In those papers families of graphs with smaller number of vertices than the sphere cardinal were given.
Specifically, one of the graph constructions in Macbeth \textsl{et al.}~\cite{Macbeth} is given for infinitely many degrees $2n$ of graphs of diameter 2 and $\frac{3}{2}(n^2-1)=\frac{3}{4}|B_2^n|-\frac{3}{2}n-\frac{9}{4}$ vertices.
Then, Vetr\'{i}k~\cite{Vetrik} constructs graphs with diameter 3 and $\frac{9}{128}(2n+3)^2(2n-5)$ vertices, which is asymptotically $\frac{27}{64}|B_3^n|$; it is remarkable that these graphs have error correction capacity 1 instead of the expected 2, and thus they do not induce quasi-perfect codes. Note that a Cayley graph attaining the degree-diameter bound will induce a perfect code and \textsl{vice versa}.

%Furthermore, the graphs considered in this paper seemed to be good expanders. Therefore, the authors computed the spectrum of some of them and the obtained values exhibit that they are Ramanujan graphs \cite{Davidoff}. Ramanujan graphs are expander graphs that attain the spectral bound. In this aim the next conjecture is formulated:
%\begin{conjecture} $\mathcal{G}_p$ is a Ramanujan graph for any prime $p \geq 19$.
%\end{conjecture}

Furthermore, the graphs considered in this paper seemed to be good expanders. Therefore, the authors computed the spectrum of some of them and the results show that they are Ramanujan graphs.
\textsl{Ramanujan graphs} are good expander graphs that attain the spectral bound~\cite{Davidoff}.
More specifically, $\mathcal G$ is a Ramanujan graph if and only if for every eigenvalue $\lambda$ of its adjacency matrix it holds either $|\lambda|=\deg(\mathcal G)$ or $|\lambda|\leq 2\sqrt{\deg(\mathcal G)-1}$.
Therefore, the following conjecture is proposed.

\begin{conjecture} $\mathcal{G}_p$ is a Ramanujan graph for any prime $p \equiv 3\pmod 4$.
\end{conjecture}

This conjecture has been verified for all primes $p< 1000$; the only primes in that range for which $\mathcal G_p$ is not Ramanujan are 17, 53 and 541.
Moreover, the authors believe that \textsl{most primes} fulfilling $p\equiv 1\pmod 4$ are such that $\mathcal{G}_p$ is also a Ramanujan graphs.
Therefore, the proof of this conjecture and the study of the relation between Golomb and Welch conjecture and spectral analysis will be considered as future work.

\section*{Acknowledgments}
The authors wish to thank the Associate Editor Professor Chaoping Xing and the anonymous referees, whose comments have improved the quality of the paper.

This work has been supported by the Spanish Science and Technology Commission (CICYT) under contracts TIN2013-46957-C2-2-P and AP2010-4900.

%\appendices
%\input{gaussian.tex}
%%\input{numb_theory.tex}

%\begin{thebibliography}{99}
%%\bibitem{ISIT10} C. Mart\'{\i}nez, C. Camarero, R. Beivide ``Perfect Graph Codes over Two Dimensional Lattices". Accepted for publication at 2010 IEEE Internations Symposium on Information Theory.
%\bibitem{Stichtenoth}Henning Stichtenoth ``Algebraic Function Fields and Codes.'' (2008 2nd ed.). Springer Publishing Company, Incorporated.
%\end{thebibliography}
\bibliographystyle{plain}
\bibliography{main}

% \begin{IEEEbiographynophoto}{Crist\'{o}bal Camarero}
% received the Master Degree in Computer Science (with distinction) in 2011 and the Ph.D. degree in 2015, both from the University of Cantabria, Spain.
% Currently, he is a researcher in the University of Cantabria.
% His research interests include graph theory with applications to interconnection networks and coding theory.
% \begin{IEEEbiographynophoto}{Carmen Mart\'{\i}nez} received the Ph.D. degree in Mathematic from the University of Cantabria (Spain) in 2007. Currently, she is an Associate Professor in the
% Computer Science and Electronics Department
% of the University of Cantabria.
% Her research interests include graph theory with application to coding theory and optimal topologies for interconnection networks.

\end{document}